\documentclass[paper=a4, 	english, 	final 	]{scrartcl}

\usepackage{babel}

\usepackage[T1]{fontenc}
\usepackage[utf8]{inputenc}

\usepackage[final, 	babel 	]{microtype}

\usepackage{amsfonts} \usepackage{amssymb}  \usepackage{mathtools} \mathtoolsset{showonlyrefs}

\usepackage{amsthm} \theoremstyle{plain}
\newtheorem{thm}[equation]{Theorem}
\newtheorem{prop}[equation]{Proposition}

\newtheorem*{thm*}{Theorem}
\newtheorem*{prop*}{Proposition}

\newtheorem*{principle*}{Principle}

\theoremstyle{definition}

\newtheorem{lemma}[equation]{Lemma}
\newtheorem{defn}[equation]{Definition}

\newtheorem*{cor*}{Corollary}
\newtheorem*{lemma*}{Lemma}
\newtheorem*{defn*}{Definition}

\theoremstyle{remark}
\newtheorem{rem}[equation]{Remark}
\newtheorem{ex}[equation]{Example}

\newtheorem*{rem*}{Remark}
\newtheorem*{ex*}{Example} \numberwithin{equation}{section}

\usepackage{tikz} \usetikzlibrary{arrows,matrix,positioning}
\tikzset{pf/.style={>=stealth,->,font=\scriptsize},
	surj/.style={->>},
	inj/.style={right hook->},
	bij/.style={above,sloped,inner sep=0.5pt},
	gl/.style={-,double},
	mat/.style={matrix of math nodes, row sep=2.5em, column sep = 2.5em, text height=1.5ex, text depth=0.25ex},
	dr/.style={matrix of math nodes, row sep=2.5em, column sep = 1.25em, text height=1.5ex, text depth=0.25ex},
	seq/.style={matrix of math nodes, row sep=2em, column sep = 2em, text height=1.5ex, text depth=0.25ex}}
\newenvironment{diag*}{\[\begin{tikzpicture}}{\end{tikzpicture}\]\ignorespacesafterend}
\newenvironment{diag}{\begin{equation}\begin{tikzpicture}[baseline=(current  bounding  box.center)]}{\end{tikzpicture}\end{equation}\ignorespacesafterend}

\usepackage{csquotes} \usepackage[backend=biber,
	style=authoryear, 	hyperref=auto, 	backref=false, 	isbn=false
	]{biblatex}
\NewBibliographyString{fromrussian}
\DefineBibliographyStrings{english}{
  fromrussian = {from the Russian}
}

\title{Super Riemann surfaces, metrics, and gravitinos}
\author{
	\texorpdfstring{
		Jürgen Jost\thanks{jjost@mis.mpg.de}
		\and Enno Keßler\thanks{kessler@mis.mpg.de (corresponding author)}
		\and Jürgen Tolksdorf\thanks{tolksdor@mis.mpg.de}
	}
	{Jürgen Jost, Enno Keßler, and Jürgen Tolksdorf}
}
\publishers{Max-Planck-Institut für Mathematik in den Naturwissenschaften}
\date{}

\usepackage{tensor}

\usepackage[final, 	pdfusetitle 	]{hyperref}
\hypersetup{colorlinks=false}

\DeclareMathOperator{\Ber}{Ber}

\DeclareMathOperator{\Div}{div}
\DeclareMathOperator{\Diff}{Diff}

\DeclareMathOperator{\GL}{GL}

\DeclareMathOperator{\id}{id}

\DeclareMathOperator{\OGL}{O}

\DeclareMathOperator{\SDiff}{SDiff}

\DeclareMathOperator{\UGL}{U}

\newcommand{\ic}{\mathrm{i}}

\newcommand{\Top}[1]{{\|#1\|}}

\newcommand{\Smooth}[1]{{|#1|}}

\renewcommand{\d}{\mathop{}\!d}

\newcommand{\Dirac}{D \hspace{-2.7mm}\slash\hspace{0.75mm}}

\newcommand{\Lie}{L}

\newcommand{\bbC}{\mathbb{C}}

\newcommand{\bbR}{\mathbb{R}}
\newcommand{\bbZ}{\mathbb{Z}}

\newcommand{\cD}{\mathcal{D}}

\newcommand{\cI}{\mathcal{I}}

\newcommand{\cO}{\mathcal{O}}

\DeclareMathOperator{\SUSY}{SUSY}

\DeclareMathOperator{\susy}{susy}

\makeatletter
\renewcommand*{\@fnsymbol}{\@arabic}
\makeatother

\begin{document}
\maketitle
\begin{abstract}
	The underlying even manifold of a super Riemann surface is a Riemann surface with a spinor valued differential form called gravitino.
	Consequently infinitesimal deformations of super Riemann surfaces are certain infinitesimal deformations of the Riemann surface and the gravitino.
	Furthermore the action functional of non-linear super symmetric sigma models, the action functional underlying string theory, can be obtained from a geometric action functional on super Riemann surfaces.
	All invariances of the super symmetric action functional are explained in super geometric terms and the action functional is a functional on the moduli space of super Riemann surfaces.
\end{abstract}

\addsec{Introduction}
Let \(\Smooth{M}\) be a compact closed two dimensional manifold.
In super string theory and super gravity one studies a super symmetric extension of the harmonic action functional where both the field \(\varphi\colon \Smooth{M}\to\bbR\) and the Riemannian metric \(g\) on \(\Smooth{M}\) get a super partner.
See for example~\cites{DZ-CASS}{BdVH-LSRIASS}.
Let \(S\) be a spinor bundle on \(\Smooth{M}\) with respect to a chosen spin structure and \(S^\vee\) its dual bundle.
Let \(\psi\) be a section of \(S^\vee\) and \(\chi\) a spinor valued differential form, i.e.\ a section of \(T^\vee\Smooth{M}\otimes_\bbR S\).
The super symmetric action functional is
\begin{multline}
	\label{AF}
	\tag{\textasteriskcentered}
	A(\varphi, \psi, g, \chi) = \int_{\Smooth{M}} \left(\vphantom{\frac12}\|\d\varphi\|_g^2 + \langle \psi, \Dirac\psi\rangle \right. \\
		\left. + 2\langle \gamma^a\gamma^b\chi_a,\psi\rangle\partial_{x^b}\varphi + \frac{1}{2}\langle\chi_a, \gamma^b\gamma^a\chi_b\rangle\langle\psi,\psi\rangle\right) \d vol_g
\end{multline}
This action is invariant under
\begin{itemize}
	\item Diffeomorphisms of \(\Smooth{M}\): \(A(\varphi\circ f, f^*\psi, f^*g, f^*\chi) = A(\varphi, \psi, g, \chi)\)
	\item Conformal transformations: \(A(\varphi, \psi, \lambda^2 g, \chi) =  A(\varphi, \psi, g, \chi)\)
	\item Super Weyl transformations: \(A(\varphi, \psi, g, \chi^k + \gamma^k s) = A(\varphi, \psi, g, \chi)\) 	\item Super symmetry:
		\begin{align}
			\delta\varphi &= \langle q, \psi \rangle & \delta\psi &= \left(\partial_{x^k}\varphi - \langle\psi, \chi_k\rangle\right)\gamma^k q\\
			\delta f_a &= -2\langle\gamma^b q, \chi(f_a)\rangle f_b & \delta\chi_a &= \nabla^S_{f_a}q
		\end{align}
		Here \(f_a\) is a \(g\)-orthonormal frame.
\end{itemize}

The aim of this paper is to clarify the relation of the action functional~\eqref{AF} to super Riemann surfaces.
The invariances of the action functional arise from geometric properties of super Riemann surfaces.
Super Riemann surfaces are an analogue of Riemann surfaces in super geometry.
This mathematical theory was developed already in the seventies for the treatment of super symmetric theories in high energy physics (see e.g.~\cites{L-ITS}{K-GMLTP}{M-GFTCG}).
The concept of a super Riemann surface appeared only a little later and their moduli space was studied, see for example~\cites{GN-GSRS}{LBR-MSRS}{CR-SRSUTT}{S-GAASTS}{DW-SMNP}.
But the precise connection between the super Riemann surfaces and the metric field \(g\) and the gravitino \(\chi\) remained unclear even though many conjectured a super Teichmüller theory that would study the moduli space of super Riemann surfaces (or a covering of it) in terms of the metric and the gravitino field.
The action functional~\eqref{AF} was claimed to arise from a particular Berezin integral on a super Riemann surface (e.g.~\cite{dHP-GSP}).
However, no explicit proof of this claim seems to exist.

In the first section of this article we introduce the concept of underlying even manifolds for families of super manifolds.
It will be shown that such an underlying even manifold \(\Smooth{M}\) exists for all super manifolds \(M\).
Any Berezin integral on \(M\) can then be reduced to an integral on \(\Smooth{M}\).

In the second section we will study the geometric structures induced on manifolds underlying super Riemann surfaces.
We will show that the geometry is completely determined by a metric \(g\) and a gravitino \(\chi\) on an underlying even manifold \(\Smooth{M}\).
This opens the possibility for a super Teichmüller theory, i.e.\ a theory of the moduli space of super Riemann surfaces in therms of metrics and gravitinos.
As a first step we study the tangent space to the moduli space of super Riemann surfaces, using metrics and gravitinos.

The aim of the third section is to demonstrate how the action functional~\eqref{AF} arises from a Berezin integral on a super Riemann surface.
The formulation in terms of the Berezin integral leads to a very clear geometrical interpretation of the symmetries of~\eqref{AF}.
Consequently the action functional~\eqref{AF} is a functional on the moduli space of super Riemann surfaces.
We give an interpretation of its energy momentum tensor and super current in terms of cotangent vectors to the moduli space.

In this paper, we present the main results of the second author's thesis (\cite{EK-DR}).
Some of the results in the last two sections rely on long and complicated computations.
In order not to overly burden the presentation we have omitted those and refer instead to the forthcoming thesis~\cite{EK-DR}.
 \section{Super Geometry}
We use the ringed space approach to super geometry (see, for example,~\cite{L-ITS}).
\begin{defn}
	A (smooth) super manifold is a locally ringed space \((\Top{M}, \cO_M)\) that is locally isomorphic to \(\bbR^{m|n} = (\bbR^m, C^\infty(\bbR^m, \bbR)\otimes_\bbR\Lambda_n)\).
	Here \(\Lambda_n\) is a real Grassmann algebra generated by \(n\) elements.
	A map of super manifolds \(f\colon M\to N\) is a map of locally ringed spaces.
	That is, a pair \((\Top{f}, f^\#)\) consisting of a continuous map \(\Top{f}\colon \Top{M}\to \Top{N}\) and a sheaf homomorphisms \(f^\#\colon \cO_N\to \cO_M\).
	It follows that the sheaf of rings \(\cO_M\) is a super commutative \(\bbZ_2\)-graded sheaf of rings.
	The elements of \(\cO_M\) will be called functions.
\end{defn}
Let \(x^a\), \(a=1,\ldots, m\) be the standard coordinate functions on \(\bbR^{m}\) and \(\eta^\alpha\), \(\alpha=1, \ldots, n\) be generators for \(\Lambda_n\).
Their lift to \(\cO_{\bbR^{m|n}}\) will be called coordinates for \(\bbR^{m|n}\).
We write \(X^A = (x^a, \eta^\alpha)\), using the convention, that small Latin letters refer to even objects, small Greek letters to odd ones and capital Latin indices refer to odd and even objects together.
Any function on \(\bbR^{m|n}\) can be expanded as
\begin{equation}
	f = \sum_{\underline{\alpha}} \eta^{\underline{\alpha}} f_{\underline{\alpha}}(x)
\end{equation}
where \(\underline{\alpha}\) is a \(\bbZ_2\)-multiindex and the \(f_{\underline{\alpha}}\) are smooth functions that can be expressed in the coordinates \(x^a\).
According to~\cite[Theorem 2.17]{L-ITS} any morphism between super domains \(U\subseteq\bbR^{m|n}\) and \(V\subseteq\bbR^{p|q}\) can be given in terms of coordinates.
\begin{ex}\label{ex:MapToR}
	Let \(X^A = (x^a, \eta^\alpha)\) be coordinates on \(\bbR^{2|2}\).
	Any map \(\varphi\colon\bbR^{2|2}\to\bbR\) is determined by the pullback of the coordinate \(r\) on \(\bbR\):
	\begin{equation}
		\varphi^\#r = f_0(x) + \eta^2\eta^1 f_{21}(x)
	\end{equation}
	Here \(f_0(x)\) and \(f_{21}(x)\) are smooth functions depending only on \(x^a\).
	Note that there is no term proportional to \(\eta^\alpha\) because the ring homomorphisms \(\varphi^\#\) preserve automatically the \(\bbZ_2\)-parity of the super functions.
\end{ex}
For the applications we have in mind one expects the full Taylor expansion.
Therefore we need to work with families of super manifolds.
\begin{defn}[\cite{L-ITS}]
	A submersion \(p_M\colon M\to B\) of super manifolds is also called a family of super manifolds over \(B\).
	A morphism \(f\) of families of super manifolds from \(p_M\colon M\to B\) to \(p_N\colon N\to B\) is a morphism \(f\colon M\to N\) such that \(p_N\circ f = p_M\).
	Any super manifold is a family over \(\bbR^{0|0}\).
	Any family is locally a projection \(\bbR^{m|n}\times B\to B\). We call \(m|n\) the dimension of the family.
\end{defn}
\begin{ex}\label{ex:TaylorExpBase}
	Consider the trivial families of super manifolds given by \(\bbR^{2|2}\times B\) and \(\bbR\times B\).
	A map \(\varphi\colon \bbR^{2|2}\times B\to \bbR\times B\) of families over \(B\) is now again given by the pullback of the coordinate function \(r\) on \(\bbR\), the map on the \(B\)-factor is determined by the properties of maps of families over \(B\).
	But this time all coefficients in the coordinate expansion can appear (using the Einstein summation convention):
	\begin{equation}
		\varphi^\#r = f_0(x) + \eta^\mu f_\mu(x) + \eta^2\eta^1 f_{21}(x)
	\end{equation}
	Here \(f_0(x)\), \(f_\mu(x)\), and \(f_{21}(x)\) are functions on \(\bbR^{2|0}\times B\).
	For all open \(U\) the ring homomorphisms  \(\left.\varphi^\#\right|_U\) must be even.
	This implies that \(f_0\) and \(f_{12}\) are even functions, whereas the functions \(f_\mu\) must be odd.
\end{ex}
\begin{lemma}[{\cite[Remark 2.6.(v)]{DM-SUSY}}]\label{lemma:Pullback}
	Let \(b\colon B'\to B\) a morphism of super manifolds and \(p_M\colon M\to B\) a family of super manifolds over \(B\).
	Then there exists a unique family of super manifolds \(p_{M'}\colon M'\to B'\) and a morphism \(p\colon M'\to M\) over \(b\).
\end{lemma}
According to Lemma~\ref{lemma:Pullback} it is not necessary to fix \(B\).
However \(B\) is always supposed to be “big enough”, see Example~\ref{ex:TaylorExpBase}.
Henceforth, all super manifolds and maps of super manifolds are implicitly to be understood as families of super manifolds and morphisms of families of super manifolds.
In particular, also \(\bbR^{m|n}\) is to be understood as the trivial family \(\bbR^{m|n}\times B\).

Many geometric concepts known from smooth manifolds carry over to families of super manifolds and are functorial under base change.
Examples such as tangent bundles, vector bundles, differential forms and Lie groups can be found in~\cite{DM-SUSY}.
A construction that has no analogue in differential geometry is that of the underlying even manifold.
\begin{defn}
	Let \(M=(\Top{M},\cO_M)\) be a family of super manifolds of dimension \(m|n\) over \(B\).
	A family of super manifolds \(\Smooth{M}=(\Top{M},\cO_{\Smooth{M}})\) of dimension \(m|0\) together with an embedding of families of super manifolds \(i\colon \Smooth{M}\to M\) that is the identity on the underlying topological space is called an underlying even manifold.
\end{defn}
\begin{lemma}[see e.g.~\cite{DM-SUSY}]
	Let \(M = (\Top{M},\cO_M)\) be a super manifold over \(\bbR^{0|0}\) of dimension \(m|n\).
There exists a unique underlying even manifold \(i\colon \Smooth{M}\to M\).
\end{lemma}
\begin{proof}
	Let \(\cI\subset\cO_M\) be the ideal sheaf of nilpotent elements.
	Then \(\Smooth{M} = (\Top{M},\cO_M/\cI)\) is a manifold of dimension \(m\).
	Furthermore the canonical projection \(i^\#\colon \cO_M\to \cO_M/\cI\) gives an embedding \(i\colon \Smooth{M}\to M\) which is the identity on the underlying topological space \(\Top{M}\).
\end{proof}
Though this concept was known for super manifolds it was to our knowledge never studied for families of super manifolds.
However in the case of families the uniqueness of the underlying manifold is lost, as is already seen in the following example:
\begin{ex}\label{ex:UnderlyingEvenManifoldRelRmn}
	Consider morphisms \(i\colon \bbR^{m|0}\times B \to \bbR^{m|n}\times B\).
	Denote the standard coordinates on \(\bbR^{m|0}\) by \(y^a\) and the standard coordinates on \(\bbR^{m|n}\) by \((x^b, \eta^\beta)\).
	Then \(i\) can be expressed in coordinates:
	\begin{align}
		i^\# x^b &= y^b & i^\# \eta^\beta &= {f(y)}^\beta
	\end{align}
	The first equation is given by the fact that \(i\) should be the identity on \(\Top{M}\), but the functions \({f(y)}^\beta\) are arbitrary odd functions on \(\bbR^{m|0}\times B\).

	It is always possible to find coordinates \((\tilde{x}^b, \tilde{\eta}^\beta)\) on \(\bbR^{m|n}\) such that \(i^\#\tilde{\eta}^\beta = 0\).
	Indeed, using the coordinate transformation
	\begin{align}
		\tilde{x}^b &= x^b & \tilde{\eta}^\beta &= -{f(x)}^\beta + \eta^\beta
	\end{align}
	assures \(i^\#\tilde{\eta}^\beta = 0\).

	There are automorphisms \(g\) of \(\bbR^{m|n}\times B\) such that \(i\circ g = i\).
	Those can best be expressed in the coordinates \(\tilde{x}^b, \tilde{\eta}^\beta\):
	\begin{align}
		g^\# \tilde{x}^b &= \tilde{x}^b + \tilde{\eta}^\mu g_\mu^b(\tilde{x},\tilde{\eta}) &
		g^\# \tilde{\eta}^\beta &= \tilde{\eta}^\mu g_\mu^\beta(\tilde{x}, \tilde{\eta})
	\end{align}
	The functions \(g_\mu^B\) are arbitrary functions on \(\bbR^{m|n}\times B\) with appropriate parity.
\end{ex}
\begin{thm}
\label{thm:ExistenceUnderlyingManifold}
	Let \(M=(\Top{M}, \cO_M)\) be a family of super manifolds over \(B\).
	Let furthermore \(\Top{U}\subseteq \Top{M}\) a subset (which could also be empty) such that there is an underlying even manifold \(\Smooth{U}\) with given embedding \(i_U \colon \Smooth{U}\to U\).
	There exists an underlying manifold \(\Smooth{M}\) and an embedding \(i\colon \Smooth{M}\to M\) such that \(\Smooth{U}\) coincides with \(\Smooth{M}\) and \(i\) with \(i_U\) over \(\Top{U}\).
\end{thm}
\begin{proof}
	From Example~\ref{ex:UnderlyingEvenManifoldRelRmn} we know that the claim is true for super domains.
	It remains to be shown that those local solutions can be put together to a global one.
	To this end cover the family \(p_M\colon M\to B\) of relative dimension by countably many adapted coordinate charts \(V_i\).
	That is we assume that each \(V_i\) can be written as a product \(V_i=W_i\times p_M(V_i)\) with coordinates \(X_i^A = (x_i^a, \eta_i^\alpha)\) on \(W_i\).
	We assume furthermore that the first \(l\) coordinate charts cover \(U\), i.e.\ it holds for \(j\leq l\) that \(\Top{V_j}\subseteq\Top{U}\) and \(\bigcup_{j=1}^{j=l}\Top{V_j} = \Top{U}\).
	We are now going to construct a family \(p_{\Smooth{M}}\colon \Smooth{M}\to B\) of relative dimension \(m|0\) and a map \(i\colon\Smooth{M}\to M\) over \(B\) by their restrictions to the open sets of the cover \(V_i\) in ascending order and glueing them together.

	If \(U=\emptyset\) choose an arbitrary embedding \(i|_{V_1}\colon \Smooth{V_1} \to V_1\) over \(p_M(V_1)\) as in Example~\ref{ex:UnderlyingEvenManifoldRelRmn}.

	Suppose now that we have a consistent structure of a \(m|0\)-dimensional manifold for \(\bigcup_{i=1}^j V_i\) together with the embedding \(i|_{\bigcup_{i=1}^j V_i}\) to the \(m|n\)-dimensional super manifold.
	Notice that \(i|_{\bigcup_{i=1}^j V_i}\) is a family of maps over \(\bigcup_{i=1}^j p_M(V_i)\).
	We need to show that it is possible to extend the manifold structure and the embedding to \(\bigcup_{i=1}^{j+1}V_i\).
	That is, we need to choose a map \(i|_{V_{j+1}}\colon \Smooth{V_{j+1}} \to V_{j+1}\) over \(p_M(V_{j+1})\) that coincides with the already constructed \(i|_{\bigcup_{i=1}^j V_i}\) on \(V_{j+1}\cap\bigcup_{i=1}^j V_i\).
	As we have seen in Example~\ref{ex:UnderlyingEvenManifoldRelRmn} such a map \(i|_{V_{j+1}}\) is given in the adapted coordinates by
	\begin{align}
		i|_{V_{j+1}}^\#x_{j+1}^a &= x^a & i|_{V_{j+1}}^\#\eta_{j+1}^\alpha &= \xi^\alpha
	\end{align}
	for some odd function \(\xi^\alpha \in\cO_{\Smooth{V_{j+1}}}\).
	By the condition that \(i|_{V_{j+1}}\) coincides with the already constructed \(i|_{\bigcup_{i=1}^j V_i}\) on \(V_{j+1}\cap\bigcup_{i=1}^j V_i\) the function \(\xi^\alpha\) is prescribed on \(V_{j+1}\cap\bigcup_{i=1}^j V_i\).
	Because we work with smooth functions into contractible spaces it is possible to construct functions \(\xi^\alpha\) that extend the given one.
	The projections \(p_{\Smooth{M}}|_{\Smooth{V_{j+1}}}\) and \(p_M|_{V_{j+1}}\) are both given by the projection to the second factor of \(V_{j+1} = W_{j+1}\times p_M(V_{j+1})\).
	Consequently, \(i|_{V_{j+1}}\) can be glued with \(i|_{\bigcup_{i=1}^j V_i}\) to give a well defined map \(i|_{\bigcup_{i=1}^{j+1} V_i}\) over \(\bigcup_{i=1}^{j+1}p_M(V_i)\).

	Continuing inductively, we construct the family \(p_{\Smooth{M}}\colon\Smooth{M}\to B\) together with \(i\colon\Smooth{M}\to M\) over \(B\) that coincides with the given \(i|_U\) on \(\Top{U}\).
\end{proof}
\begin{rem}
	In the construction presented above, also the underlying manifold depends on the choices made.
	Different underlying manifolds might not only differ in the embedding \(i\colon \Smooth{M}\to M\), but also in the manifold structure of \(\Smooth{M}\).
	This can be seen by considering the induced coordinate changes on \(\Smooth{M}\).
	Let \(X^A = (x^a, \eta^\alpha)\) and \(Y^B = (y^b, \theta^\beta)\) two different adapted coordinate systems on \(V\times B\). The coordinate change is given by
	\begin{align}
		y^b &= \sum_{\underline{\gamma}}\eta^{\underline{\gamma}}{f(x)}^b_{\underline{\gamma}} & \theta^\beta = \sum_{\underline{\gamma}}\eta^{\underline{\gamma}}{f(x)}^\beta_{\underline{\gamma}}.
	\end{align}
	Let \(i\colon \Smooth{V}\to V\) be given by
	\begin{align}
		i^\#x^a &= x^a & i^\#\eta^\alpha = \xi^\alpha.
	\end{align}
	In the second coordinates \((y^b,\theta^\beta)\) the embedding \(i\) is given as
	\begin{align}
		i^\#y^b &= y^b = \sum_{\underline{\gamma}}\xi^{\underline{\gamma}}{f(x)}^b_{\underline{\gamma}} & i^\#\theta^\beta = \sum_{\underline{\gamma}}\xi^{\underline{\gamma}}{f(x)}^\beta_{\underline{\gamma}}.
	\end{align}
	As a consequence the coordinate change on \(\Smooth{V}\) is
	\begin{align}
		y^b = \sum_{\underline{\gamma}}\xi^{\underline{\gamma}}{f(x)}^b_{\underline{\gamma}}.
	\end{align}
	Therefore the coordinate changes explicitly depend on the chosen embedding \(i\), resp. \(\xi^\alpha\).
\end{rem}

The theory of integration for families of super manifolds is sketched in~\cite[§3.10]{DM-SUSY}.
For fiberwise compact, oriented families of super manifolds integration is an \(\cO_B\)-linear functional
\begin{equation}
	\int_{M}\colon \Ber T^\vee M\to \cO_B
\end{equation}
from the Berezinian of the cotangent bundle to the functions on \(B\).
The Berezinian is the generalization of the determinant bundle to super geometry.
Integration is given in local coordinates \((x^a,\eta^\alpha)\) by
\begin{equation}\label{eq:IntCoord}
	\int_{\bbR^{m|n}} g(x,\eta)[\d x^1\ldots\d x^m\d\eta^1\ldots\d\eta^n]  = \int_{\bbR^{m|0}} g_{top}(x) \d x^1\ldots\d x^m
\end{equation}
where \(g_{top}\) is the coefficient of \(\eta^n\cdot\cdots\cdot\eta^1\) of in the coordinate expansion of \(g\).
\begin{prop}\label{prop:Integration}
	Let \(i\colon \Smooth{M}\to M\) be the embedding of an underlying even manifold for a family \(M\) of fiberwise compact, orientable super manifolds over \(B\). For any section \(b\) of \(\Ber T^\vee M \) there exists a top form \(|b|\) on \(\Smooth{M}\) such that
	\begin{equation}
		\int_M b = \int_{\Smooth{M}} |b|.
	\end{equation}
\end{prop}
\begin{proof}
	The formula~\eqref{eq:IntCoord} already gives the local definition of \(|b|\) in coordinates.
	We only have to show that it transforms under coordinate change as expected.
	We can restrict our attention to the case of coordinates \(X^A = (x^a, \eta^\alpha)\) and \(Y^B = (y^b, \theta^\beta)\) such that \(i^\#\eta^\alpha = 0\) and \(i^\#\theta^\beta = 0\).
	Then the change of coordinates is given to lowest order in \(\theta\) by
	\begin{align}
		x^a &= f^a(y) + \theta\dots & \eta^\alpha &= \theta^\mu f^\alpha_\mu(y) + \theta^2\dots
	\end{align}
	It follows that
	\begin{equation}
		\Ber\frac{\partial X^A}{\partial Y^B} = \det\frac{\partial x^a}{\partial y^b}{\left(\det\frac{\partial \eta^\alpha}{\partial \theta^\beta}\right)}^{-1} + \theta\dots.
	\end{equation}
	As in~\cite{L-ITS} it is sufficient to consider \(b=\eta^n\cdot\dots\cdot\eta^1 g(x) [\d X]\). Hence
	\begin{multline}
		b = \eta^n\cdot\dots\cdot\eta^1 g(x) [\d X] \\
		= \theta^n\cdot\dots\cdot\theta^1 \left(\det f^\alpha_\mu(y)\right) g(x(y)) \left(\det\frac{\partial x^a}{\partial y^b}\right){\left(\det\frac{\partial \eta^\alpha}{\partial \theta^\beta}\right)}^{-1} [\d Y] \\
		= \theta^n\cdot\dots\cdot\theta^1 \left(\det\frac{\partial x^a}{\partial y^b}\right) g(x(y)) [\d Y]
	\end{multline}
	as expected.
\end{proof}
 \section{Super Riemann surfaces}
\label{sec:SRS}

\begin{defn}[see~\cite{LBR-MSRS}]
	A super Riemann surface is a \(1|1\)-dimensional complex super manifold \(M\) with a \(0|1\)-dimensional distribution \(\cD\subset TM\) such that the commutator of vector fields induces an isomorphism
	\begin{equation}
		\frac12 [\cdot, \cdot]\colon \cD\otimes_\bbC\cD \to TM/\cD.
	\end{equation}
\end{defn}
\begin{ex}\label{ex:SRS}
	Let \((z,\theta)\) be the standard coordinates on \(\bbC^{1|1}\) and define \(\cD\subset T\bbC^{1|1}\) by \(\cD=\langle\partial_\theta + \theta\partial_z\rangle\). The isomorphism \(\cD\otimes\cD \simeq TM/\cD\) is explicitly given by
\begin{equation}
		[\partial_\theta + \theta\partial_z, \partial_\theta + \theta\partial_z] = 2\partial_z
	\end{equation}
	This example is generic since any super Riemann surface is locally of this form, see~\cite[Lemma 1.2]{LBR-MSRS}.
\end{ex}
\begin{thm}[\cite{GN-GSRS}]\label{thm:SRSReductionOfStructureGroup}
	A super Riemann surface is a \(2|2\)-dimensional real super manifold with a reduction of the structure group to
	\begin{equation}
		G = \left\{
			\begin{pmatrix}
				A^2 & B\\
				0 & A \\
			\end{pmatrix}
		\middle| A, B\in\bbC
		\right\} \subset \GL_\bbC(1|1) \subset \GL_\bbR(2|2)
	\end{equation}
	together with the following integrability conditions.
	Remember that \(\bbC\) is to be understood as the trivial family \(\bbC\times B\).
	Denote the \(G\)-frames by \(F_z\) and \(F_+\). Their decomposition in real and imaginary part yields frames \(F_a, F_\alpha\) for \(a=1,2, \alpha=3,4\) as follows:
	\begin{align}
		F_z &= \frac12\left(F_1 - \ic F_2\right) & F_+ &= \frac12\left(F_3 - \ic F_4\right) \\
		F_{\overline{z}} &= \overline{F_z} & F_- &= \overline{F_+}.
	\end{align}
	Let us denote the structure coefficients by \(t_{AB}^C\):
	\begin{equation}
		[F_A, F_B]=t_{AB}^C F_C
	\end{equation}
	Then the integrability conditions in terms of the complex frames are given by the following \(G\)-invariant equations:
	\begin{gather}
	\label{eq:FrameIntegrabilityConditions}
		t_{z+}^{\overline{z}} = t_{z+}^- = t_{++}^{\overline{z}} = t_{++}^- = t_{+-}^z = t_{+\overline{z}}^z = 0 \\
		t_{++}^z = 2
	\end{gather}
	The vanishing of the first four structure coefficients guarantees an integrable holomorphic structure, and the vanishing of the last two that \(\cD\) is a holomorphic distribution.
	Furthermore, \(t_{++}^z=2\) gives the complete non-integrability of \(\cD\).
\end{thm}
Theorem~\ref{thm:SRSReductionOfStructureGroup} leads to two observations.
First, since \(\OGL(2|2)\nsubseteq G\) it is not possible to describe the geometry of super Riemann surfaces in terms of super Riemannian metrics on \(M\).
Second, a further reduction to \(\UGL(1)\) is always possible via
\begin{equation}
	\begin{split}
		\UGL(1)&\to G\\
		U &\mapsto
			\begin{pmatrix}
				U^2 & 0 \\
				0 & U \\
			\end{pmatrix}.
	\end{split}
\end{equation}

Consider now such a \(\UGL(1)\)-structure on \(M\).
It induces a non-degenerate, super symmetric bilinear form \(m\) on \(TM\), given in the \(\UGL(1)\)-frames by
\begin{align}
\label{eq:U1Metric}
	m(F_a, F_b) &= \delta_{ab} & m(F_a, F_\beta) &= 0 & m(F_\alpha, F_\beta) &= \varepsilon_{\alpha\beta}
\end{align}
The projector on \(\cD\) gives a splitting of the following short exact sequence:
\begin{diag}\label{seq:SRSsplit}
	\matrix[seq](m) { 0 & \cD & TM=\cD^\perp\oplus \cD & TM/\cD & 0\\ };
	\path[pf]	(m-1-1) edge (m-1-2)
		(m-1-2) edge (m-1-3)
		(m-1-3) edge (m-1-4)
		(m-1-4) edge (m-1-5)
			edge[bend right=40] node[auto]{\(p\)} (m-1-3);
\end{diag}
The pullback of the short exact sequence~\eqref{seq:SRSsplit} along an embedding \(i\colon \Smooth{M}\to M\)
\begin{diag}\label{eq:SRS_pullback_seq}
	\matrix[seq](m) { 0 & S & i^*TM & T\Smooth{M} & 0\\};
	\path[pf]	(m-1-1) edge (m-1-2)
		(m-1-2) edge (m-1-3)
		(m-1-3) edge (m-1-4)
		(m-1-4) edge (m-1-5)
			edge[bend right=40] node[auto]{\(\tilde{p}\)} (m-1-3)
			edge[bend left=40] node[auto]{\(Ti\)} (m-1-3);
\end{diag}
possesses a second splitting given by \(Ti\).
\begin{defn}[Metric \(g\), spinor bundle \(S\) and gravitino \(\chi\)]
	By the identification \(T\Smooth{M}=i^*\cD^\perp\), the tangent bundle of \(\Smooth{M}\) gets equipped with a metric \(g\).

	The bundle \(S=i^*\cD\) is a spinor bundle of the metric \(g\) because \(i^*\cD\otimes_\bbC i^*\cD = i^*TM/\cD = T\Smooth{M}\).
	The identification \(S=i^*\cD\) induces a non degenerate bilinear form \(g_S\) on \(S\) that is given in the frames \(s_\alpha = i^*F_\alpha\) by
	\begin{equation}
		g_S(s_\alpha, s_\beta) = \varepsilon_{\alpha\beta}
	\end{equation}

	The difference of the splittings \(\tilde{p}\) and \(Ti\) is a section of \(T^\vee\Smooth{M}\otimes S\) which we call gravitino \(\chi\).
	\begin{equation}\label{eq:DefinitionGravitino}
		\chi(v) = p_S\left(\tilde{p} - Ti\right)v
	\end{equation}
	Here \(p_S\colon i^*TM\to S\) is the projector given by the splitting of the short exact sequence by \(\tilde{p}\).
\end{defn}
Keep in mind that the vector bundle \(S\) is of real dimension \(0|2\) and the frames \(s_\alpha\) are odd.
Also notice that in general the embedding \(i\colon \Smooth{M}\to M\) is not holomorphic with respect to the complex structure on \(\Smooth{M}\) induced by \(g\) (cf.\ the construction in the proof of the Theorem~\ref{thm:ExistenceUnderlyingManifold}).

Different choices of \(\UGL(1)\)-structure lead to metrics and gravitinos which differ from \(g\) and \(\chi\) only by a conformal and super Weyl transformation.
Every matrix of \(G\) can be decomposed as
\begin{equation}\label{eq:GMatrixDecomposition}
	\begin{pmatrix}
		A^2 & B \\
		0 & A
	\end{pmatrix}
	=
	\begin{pmatrix}
		U^2 & 0 \\
		0 & U
	\end{pmatrix}
	\begin{pmatrix}
		R^2 & 0 \\
		0 & R
	\end{pmatrix}
	\begin{pmatrix}
		1 & T \\
		0 & 1
	\end{pmatrix}
\end{equation}
where \(U\in \UGL(1)\), \(R\in \bbR^+\) and \(T\in \bbC\).
The first matrix preserves the \(\UGL(1)\)-structure on \(M\).
Consequently the bilinear forms \(m\), \(g\) and \(g_S\) are preserved.
The second matrix in the decomposition~\eqref{eq:GMatrixDecomposition} rescales the frames \(F_A\) and changes the \(\UGL(1)\) structure.
As a result the bilinear form \(g\) is rescaled by \(i^\#R^2\) and \(g_S\) is rescaled by \(i^\#R\).
The third matrix in the decomposition changes the splitting \(TM=\cD\oplus\cD^\perp\).
It is easy to see that the induced change on \(\chi\) is indeed a super Weyl transformation.
However only the functions \(i^\#U\), \(i^\#R\) and \(i^\#T\) do effect the metric \(g\) and the gravitino.
The higher order terms of \(R\) and \(T\) do leave \(g\) and \(\chi\) invariant.

Having constructed a metric and gravitino on a 2-dimensional surface \(\Smooth{M}\) from a super Riemann surface \(M\), we now consider the opposite question.
Given a \(2|0\)-dimensional manifold \(\Smooth{M}\) and a metric \(g\) and a gravitino \(\chi\) is there a unique super Riemann surface \(M\) with an embedding \(i\colon \Smooth{M}\to M\) such that the above construction gives the same metric and gravitino back?
In order to affirmatively answer the question, one has to take into account all geometrical degrees of freedom on \(M\) that are not fixed by the metric \(g\) and the gravitino \(\chi\) on \(\Smooth{M}\).
An example for such geometrical degree of freedom is given by the higher order terms in the decomposition~\eqref{eq:GMatrixDecomposition}.
\begin{defn}[Wess--Zumino frames]
	A \(G\)-frame \(F_A\) is called Wess--Zumino frame if the following commutator relations hold in addition to the integrability conditions~\eqref{eq:FrameIntegrabilityConditions}:
	\begin{align}
	\label{eq:WZFrameCond}
		i^\#t_{+-}^+ &= 0 & i^\#F_+t_{+-}^+ &= 0 & t_{++}^+ &= 0
	\end{align}
\end{defn}
\begin{lemma}\label{lemma:ExistenceWZFrames}
	Let \(F_A\) be a \(\UGL(1)\)-frame.
	Then there is a unique Wess--Zumino frame \(\tilde{F}_A\) in the same \(G\)-class such that \(i^*F_A = i^*\tilde{F}_A\).
\end{lemma}
\begin{proof}
	Apply a transformation \(h\in G\) to \(F_A\) such that \(i^*h=\id\).
	Then the conditions~\eqref{eq:WZFrameCond} fix the higher order terms of \(U\), \(R\) and \(T\) from~\eqref{eq:GMatrixDecomposition}.
\end{proof}
\begin{defn}[Wess--Zumino coordinates]
	The coordinates \(X^A = (x^a, \eta^\alpha)\) are called Wess--Zumino coordinates of the frame \(F_A\) if \(i^\#\eta^\alpha = 0\) and the coordinate expression of the frame \(F_\alpha\) is given by
	\begin{equation}
	\label{eq:WZCoordinateConditions}
		F_\alpha = \left(\eta^\mu \tensor{F}{_\mu_\alpha^b}(x) + \eta^2\eta^1\dots \right)\partial_{x^b} + \left(\delta_\alpha^\beta + \eta^\mu\tensor{F}{_\mu_\alpha^\beta} + \eta^2\eta^1\dots \right)\partial_{\eta^\beta}.
	\end{equation}
	Here the degree one coefficients are symmetric with respect to the lower indices, i.e.
	\begin{align}
		\varepsilon^{\mu\alpha} \tensor{F}{_\mu_\alpha^C} = 0.
	\end{align}
\end{defn}
\begin{lemma}\label{lemma:ExistenceWZCoordinates}
	Given a \(G\)-frame \(F_A\) and coordinates \(\tilde{X}^A = (\tilde{x}^a, \tilde{\eta}^\alpha)\) there are unique Wess--Zumino coordinates \(X^A = (x^a, \eta^\alpha)\) for \(F_A\) such that \(i^\#\tilde{x}^a = i^\#x^a\).
\end{lemma}
The notions of “Wess--Zumino frames” and “Wess--Zumino coordinates” are derived from the notion of “Wess--Zumino gauge” as used in~\cite{dHP-GSP}.
They have at least two purposes.
The first one is that they reduce the freedom in the local description of super Riemann surfaces.
Instead of all super coordinate systems and all \(G\)-frames, we now only need to consider the Wess--Zumino frames and Wess--Zumino coordinates.
As was shown in Lemma~\ref{lemma:ExistenceWZFrames} and Lemma~\ref{lemma:ExistenceWZCoordinates} they are unique up to a choice of \(i^\#x^a\) and \(i^*F_\alpha\).
Second they relate the odd coordinates on \(M\) to spinors on \(\Smooth{M}\), as the frames \(s_\alpha = i^*F_\alpha = i^*\partial_{\eta^\alpha}\) are frames for \(S\).

Let now \(F_A\) be a \(\UGL(1)\)-frame on \(U\subset M\).
Consider the coordinate expansion in Wess--Zumino coordinates \(X^A = (x^a, \eta^\alpha)\) for \(F_A\).
\begin{equation}\label{eq:FrameExpansion}
	\begin{split}
		F_a &= \left(\tensor{F}{_0_a^b} + \eta^\mu\tensor{F}{_\mu_a^b} + \eta^2\eta^1\tensor{F}{_{21}_a^b}\right)\partial_{x^b} + \left(\tensor{F}{_0_a^\beta} + \eta^\mu\tensor{F}{_\mu_a^\beta} + \eta^2\eta^1\tensor{F}{_{21}_a^\beta}\right) \partial_{\eta^\beta} \\
		F_\alpha &= \left(\eta^\mu\tensor{F}{_\mu_\alpha^b} + \eta^2\eta^1\tensor{F}{_{21}_\alpha^b}\right)\partial_{x^b} + \left( \delta_\alpha^\beta + \eta^\mu\tensor{F}{_\mu_\alpha^\beta} + \eta^2\eta^1\tensor{F}{_{21}_\alpha^\beta}\right) \partial_{\eta^\beta}
	\end{split}
\end{equation}
The frames \(s_\alpha = i^*F_\alpha\) are \(\UGL(1)\)-frames for \(S\).
Furthermore the frame \(F_a\) can be expanded
\begin{equation}
	i^*F_a = \tensor{F}{_0_a^b}i^*\partial_{x^b} + \tensor{F}{_0_a^\beta}s_\beta
\end{equation}
Then by formula~\eqref{eq:DefinitionGravitino} we know that \(f_a = \tensor{F}{_0_a^b}\partial_{y^b}\) is a \(g\)-orthonormal frame and the gravitino is given by
\begin{equation}
	\chi(f_a) = \tensor{F}{_0_a^\beta}s_\beta
\end{equation}
To complete the local description of super Riemann surfaces in terms of metrics and gravitinos we still need the following lemma:

\begin{lemma}\label{lemma:WZPair}
	Let \(F_A\) be a Wess--Zumino frame and \(X^A = (x^a, \eta^\alpha)\) Wess--Zumino coordinates for \(F_A\).
	All higher order coefficients in~\eqref{eq:FrameExpansion} can be expressed in terms of \(\tensor{F}{_0_a^b}\) and \(\tensor{F}{_0_a^\beta}\) and thus in terms of \(f_a\) and \(\chi\).
\end{lemma}
\begin{proof}
	The equations~\eqref{eq:FrameIntegrabilityConditions},~\eqref{eq:WZCoordinateConditions} and~\eqref{eq:WZFrameCond} are solvable for the unknown coefficient functions.
\end{proof}

\begin{thm}
	Given a super manifold \(\Smooth{M}\) over \(B\) together with a metric \(g\), a spinor bundle \(S\) and a gravitino field \(\chi\). Then there is a unique super Riemann surface \(M\) over \(B\) together with an inclusion \(i\colon \Smooth{M}\to M\) such that the above procedure gives back the gravitino and metric up to conformal transformation of \(g\) and super Weyl transformation of the gravitino \(\chi\).
\end{thm}
\begin{proof}
	Cover \(\Smooth{M}\) by open coordinate sets \((V, y^a)\).
	Choose local \(\UGL(1)\)-frames \(s_\alpha\) of \(S\) and \(f_a\) of \(T\Smooth{M}\) such that \(s_+\otimes_\bbC s_+\mapsto f_z\).
	Construct over the topological space \(V\) the super manifold \((V,\cO_V)\) by setting \(\cO_V=\Lambda(\Gamma_V(S^\vee))\) with coordinates \(x^a=y^a\) and \(\eta^\alpha = s^\alpha\), where \(s^\alpha\) is the canonical dual basis to \(s_\alpha\).
	Denote by \(F_A\) the Wess--Zumino frame constructed from the coefficients of the frame \(f_a\) and the gravitino \(\chi\) according to Lemma~\ref{lemma:WZPair}.
	This gives an integrable \(G\)-reduction of the structure group of \(TV\).
	The map \(i\) is locally constructed via its action on the coordinates \(x^a, \eta^\alpha\).

	It remains to glue different local constructions in order to obtain a well defined super Riemann surface over the same topological space \(\Top{M}\).
	The Wess--Zumino frames over different trivializing covers may differ by a \(\UGL(1)\)-transformation.
	The Wess--Zumino coordinates of Wess--Zumino frames that differ by a \(\UGL(1)\)-transformation are completely fixed by this \(\UGL(1)\)-transformation.
	Once more, details may be found in~\cite{EK-DR}.
\end{proof}

We have shown a one to one correspondence
\begin{equation}
	\left\{i\colon \Smooth{M}\to M, M \text{ super Riemann surface}\right\} \longleftrightarrow \left\{ \Smooth{M}, S, g, \chi \right\} / \text{Weyl, SWeyl}
\end{equation}
An advantage of this description is that on the right-hand side there are no integrability conditions to be fulfilled.
On the left hand side the integrability conditions~\eqref{eq:FrameIntegrabilityConditions} have to be fulfilled.
The presence of the integrability conditions complicates the study of deformations as one needs to assure the integrability of the deformations.

To obtain a super Teichmüller description of the moduli space of super Riemann surfaces one may look for a one to one correspondence (see e.g.~\cite{JJ-GP} and references therein)
\begin{multline}
	\left\{M, M \text{ super Riemann surface}\right\}/\SDiff(M) \\
	\longleftrightarrow \\
	\left\{ \Smooth{M}, S, g, \chi\right\} / \text{Weyl, SWeyl,} \Diff(\Smooth{M}), \SUSY
\end{multline}
The super symmetry transformations \(\SUSY\) on the right hand side can probably be identified with the change of embedding \(i\).
A precise definition of \(\SUSY\) and the study of the full quotient must be left for further research.
As a first step we treat the infinitesimal case.
As a preparation we first study the infinitesimal change of embedding.
\begin{prop}
	Let \(i_t\colon \Smooth{M}_t\to M\) be a smooth family of embeddings. The infinitesimal deformation
	\begin{equation}
		\left.\frac{\d}{\d t}\right|_{t=0} i_t \in \Gamma_{\Smooth{M}_0}(i_0^*TM)
	\end{equation}
	is a section \(q\in\Gamma_{\Smooth{M}_0}(i_0^*\cD)\) and the derivatives of the families of local frames \({f(t)}_a\) and \({\chi(t)}_a\) are given by
	\begin{equation}\label{eq:MetricGravitinoSUSY}
		\begin{split}
			{\left.\frac{\d}{\d t}\right|}_{t=0} {f(t)}_a &= -2\langle\gamma^b q, \chi(f_a)\rangle f_b \\
			{\left.\frac{\d}{\d t}\right|}_{t=0} {\chi(t)}_a &= \nabla_{f_a}^S q = \nabla^{LC}_{f_a} q + \langle \gamma^b\chi_b, \chi_a\rangle \gamma^1\gamma^2q
		\end{split}
	\end{equation}
	Here \(\nabla^{LC}\) is the Levi-Civita connection lifted to \(S\).
\end{prop}
This proposition also justifies that the field \(\chi\) defined above was called gravitino, because the transformations~\eqref{eq:MetricGravitinoSUSY} are the expected super symmetries. Compare~\cites{dHP-GSP}{JJ-GP}{BdVH-LSRIASS}.

\begin{lemma}
	The gravitino can be gauged to zero locally.
	More precisely for every point \(m\in\Top{M}\) there exists an open neighbourhood \(U\subseteq M\) such that there is a \(\UGL(1)\)-structure and an embedding \(i\colon |M|\to M\) such that \(\left.\chi\right|_{i^{-1}(U)} = 0\).
	The gravitino can be gauged away globally, if \(M\) is a trivial family of super Riemann surfaces.
\end{lemma}
\begin{proof}
	Choose around \(m\) complex coordinates \((z,\theta)\) such that \(\cD=\langle\partial_\theta + \theta\partial_z\rangle\) (see example~\ref{ex:SRS}).
	Let the \(\UGL(1)\)-structure be given by the frames \(F_z = \partial_z\) and \(F_+ = \partial_\theta + \theta\partial_z\) and the embedding by \(i^\#\theta = 0\).
	Then the gravitino vanishes on \(U\).
\end{proof}

\begin{thm}
	The infinitesimal deformations of a super Riemann surface \(M\) with embedding \(i\colon \Smooth{M}\to M\) given by \(g\) and \(\chi\) on \(\Smooth{M}\) are given by
	\begin{equation}
		H^0(T^\vee\Smooth{M}\otimes_\bbC T^\vee\Smooth{M})\oplus H^0(S^\vee\otimes_\bbC S^\vee\otimes_\bbC S^\vee)
	\end{equation}
	Here \(H^0\) denotes holomorphic sections.
\end{thm}
\begin{proof}
	Let the super Riemann surface \(M\) be given in by a metric \(g\) and a gravitino \(\chi\) on \(\Smooth{M}\) with embedding \(i\colon \Smooth{M}\to M\).
	Any infinitesimal deformation is given by an infinitesimal deformation of metric \(h\) and gravitino \(\rho\).
	However not every infinitesimal deformation of metric and gravitino give rise to an infinitesimal deformation of the super Riemann surface.
	The infinitesimal deformations of the metric and gravitino induced by Weyl and super Weyl, diffeomorphisms, and super symmetry do lead to equivalent super Riemann surfaces.
	We will thus need to decompose the infinitesimal deformation \(h\) of the metric as
	\begin{equation}\label{eq:DecompositionOfMetric}
		h = \lambda g + \Lie_X g + \susy(q) + D
	\end{equation}
	for some infinitesimal Weyl transformation with parameter \(\lambda\), a Lie derivative along the vector field \(X\) and infinitesimal super symmetry transformation \(\susy(q)\) given by the spinor \(q\) as in~\eqref{eq:MetricGravitinoSUSY}.
	The parameters \(\lambda\), \(X\), and \(q\) need to be determined.
	The remaining part \(D\) is a true even infinitesimal deformations of the super Riemann surface.
	Analogously, the infinitesimal deformation \(\rho\) of the gravitino needs to be decomposed in
	\begin{equation}\label{eq:DecompositionOfGravitino}
		\rho = \gamma t + \Lie_X \chi + \nabla^S q + \mathfrak{D}
	\end{equation}
	for some spinor \(t\), that rests to be determined.
	The remaining part \(\mathfrak{D}\) is a true odd infinitesimal deformation of the super Riemann surface.

	We will work in local holomorphic coordinates \(z=x^1+\ic x^2\) defined in some open neighbourhood \(U\).
	First we consider the special case \(g_{ij} = \delta_{ij}\) and \(\chi = 0\).
	Let \(X=X^k\partial_{x^k}\).
	The equation~\eqref{eq:DecompositionOfMetric} simplifies to
	\begin{equation}
		h_{ij} = \lambda \delta_{ij} + \left(\partial_{x^i}X^k\right)\delta_{kj} + \left(\partial_{x^j}X^k\right)\delta_{ki} + D_{ij}
	\end{equation}
	Letting
	\begin{equation}
		\lambda = \frac12 h_{ij}\delta^{ij} - \left(\partial_{x^k}X^k\right)
	\end{equation}
	it is possible to assume \(D_{ij}\) is symmetric and trace free.
	As a consequence, the bilinear form \(D\) can be identified with a section of \(T^\vee\Smooth{M}\otimes_\bbC T^\vee\Smooth{M}\).
	\begin{equation}
		\begin{pmatrix}
			a & b \\
			b & -a
		\end{pmatrix}
		\mapsto \left(a-\ic b\right)\d{z}\otimes\d{z}
	\end{equation}
	It is possible to choose the vector field \(X\) such that \(D\) is a holomorphic quadratic differential.
	The holomorphicity condition for \(D\) is equivalent to the following Laplace equations for \(X^k\):
	\begin{align}
		0 = \partial_{x^1}a + \partial_{x^2}b &= \frac12\partial_{x^1}\left(h_{11} - h_{22}\right) + \partial_{x^2}h_{12} - \partial_{x^1}^2X^1 - \partial_{x^2}^2X^1\\
		0 = -\partial_{x^1}b + \partial_{x^2}a &= \frac12\partial_{x^2}\left(h_{11} - h_{22}\right) - \partial_{x^1}h_{12} + \partial_{x^1}^2X^2 + \partial_{x^2}^2X^2\\
	\end{align}
	We have decomposed every infinitesimal deformation \(h\) of the metric \(g\) into an infinitesimal Weyl transformation, a Lie derivative and a holomorphic quadratic differential.
	The holomorphic quadratic differentials represent the true even deformations of \(M\).

	In an analogous manner we are going to proceed with the deformation \(\rho\) of the gravitino.
	It will be convenient to consider \(\rho\) as a section of \(T^\vee\Smooth{M}\otimes S^\vee\).
	We choose a complex basis \(s_+\) for \(S\) such that \(s_+\otimes s_+ = \partial_z\), and let \(s_+=s_3 - \ic s_4\).
	The corresponding dual basis will be denoted \(s^+\) and \(s^\alpha\) respectively.
	The vector bundle \(T^\vee\Smooth{M}\otimes S^\vee\) can be decomposed in \(S^\vee\oplus S^\vee\otimes_\bbC S^\vee\otimes_\bbC S^\vee\).
	In the basis we use here the spinor part of an arbitrary section \(\rho\) is given by \(s^\alpha{\gamma^a}_\alpha^\beta\rho_{a\beta}\).
	The equation~\eqref{eq:DecompositionOfGravitino} is given in our local coordinates by
	\begin{equation}
		\rho_{a\beta} = \delta_{ab}{\gamma^a}_\beta^\mu\varepsilon_{\mu\nu}t^\nu - \varepsilon_{\beta\mu}\left(\partial_{x^a}q^\mu\right) + \mathfrak{D}_{a\beta}
	\end{equation}
	It is possible to fix the spinor \(t\) such that \(\mathfrak{D}\) is in \(S^\vee\otimes_\bbC S^\vee\otimes_\bbC S^\vee\), i.e.
	\begin{equation}
		0 = {\gamma^a}_\alpha^\beta\left(\rho_{a\beta} + \varepsilon_{\beta\mu}\partial_{x^a}q^\mu\right)  - 2\varepsilon_{\alpha\nu}t^\nu
	\end{equation}
	Consequently the coefficients of \(\mathfrak{D}\) fulfil
	\begin{align}
		\mathfrak{D}_{13} + \mathfrak{D}_{24} &= 0 & \mathfrak{D}_{23} - \mathfrak{D}_{14} &= 0
	\end{align}
	The cospinor valued differential form \(\mathfrak{D}\) can be identified with
	\begin{equation}
		\left(\mathfrak{D}_{13} + \ic\mathfrak{D}_{14}\right) \d{z}\otimes s^+
	\end{equation}
	The condition, for \(\mathfrak{D}\) to be a holomorphic section of \(S^\vee\otimes_\bbC S^\vee\otimes_\bbC S^\vee\) is given again by the Cauchy--Riemann equations
	\begin{align}
		0 = \partial_{x^1}\mathfrak{D}_{13} - \partial_{x^2}\mathfrak{D}_{14} &= \frac12\left(\partial_{x^1}\left(\rho_{13} - \rho_{24}\right) + \partial_{x^2} \left(\rho_{14} + \rho_{23}\right)\right) + \partial_{x^1}^2q^4 + \partial_{x^2}^2q^4 \\
		0 = \partial_{x^2}\mathfrak{D}_{13} + \partial_{x^1}\mathfrak{D}_{14} &= \frac12\left(\partial_{x^2}\left(\rho_{13} - \rho_{24}\right) - \partial_{x^1} \left(\rho_{14} + \rho_{23}\right)\right) - \partial_{x^1}^2q^3 - \partial_{x^2}^2q^3
	\end{align}
	We can thus decompose the infinitesimal deformations of the gravitino in an infinitesimal super Weyl transformation, an infinitesimal super symmetry, and a holomorphic section of \(S^\vee\otimes_\bbC S^\vee\otimes_\bbC S^\vee\).

	In the general case \(\chi\neq 0\) the Cauchy--Riemann equations for \(D\) and \(\mathfrak{D}\) are given by a system of coupled partial differential equations for \(X\) and \(q\) of generic form:
	\begin{align}
		\Delta X &= f(X, X', q, q') & \Delta q &= g(X, X', X'', q, q')
	\end{align}
	The system of Cauchy--Riemann equations can thus be solved by the theory of elliptic partial differential equations, since the Cauchy--Riemann equations for \(D\) do not contain second derivatives of \(q\).
	The reason is that \(\susy(q)\) does not contain derivatives of~\(q\).
\end{proof}
Similar statements for trivial families can be found in~\cite{LBR-MSRS} or~\cite{S-GAASTS}.
However the version given here is more general, as it allows for non-trivial families.
Furthermore the proof given here shows directly which deformations of metric and gravitino correspond to infinitesimal deformations of the given super Riemann surface.

The complex dimension of the infinitesimal deformation space can be calculated by the theorem of Riemann-Roch in the case of \(B=\bbR^{0|0}\).
The dimension is found to be \(3p-3|2p-2\) for genus \(p\geq 2\).
 \section{The action functional}
We now turn again to the action functional~\eqref{AF}.
In this section we assume that \(M\) is a fiberwise compact family of super Riemann surfaces with a compatible super metric \(m\).
Let \(N\) be an arbitrary (super) manifold with Riemannian metric \(n\) and Levi-Civita covariant derivative \(\nabla^{TN}\).
For details on Levi-Civita covariant derivatives on super manifolds see~\cite{G-RSG}.
Consider a morphism \(\Phi\colon M\to N\).
The action
\begin{equation}
\label{eq:AFSRS}
	A(M, \Phi) = \frac12\int_M \|\left.T \Phi\right|_{\cD}\|^2_{\left.m^\vee\right|_{\cD^\vee}\otimes \Phi^*n} [\d{vol_m}]
\end{equation}
might be seen as a generalization of the harmonic action functional to super Riemann surfaces.
Remark that in contrast to the harmonic action functional the tangent map \(T\Phi\) is restricted to the subbundle \(\cD\) in \(TM\).
Given \(\UGL(1)\)-frames \(F_A\) the action can be written as
\begin{equation}
\label{eq:AFSRSFrames}
	A(M, \Phi) = \frac12\int_M \varepsilon^{\alpha\beta} \langle F_\alpha\Phi, F_\beta\Phi\rangle_{\Phi^*n} [F^1 F^2 F^3 F^4].
\end{equation}
The action~\eqref{eq:AFSRS} can be found in different forms in the literature, see in particular~\cites{GN-GSRS}{dHP-GSP}.
In~\cite{GN-GSRS} one can find an explicit proof for the \(G\)-invariance of~\eqref{eq:AFSRSFrames}.
Thus the action functional does not depend on the metric \(m\), but rather only on the super Riemann surface structure, i.e.\ the \(G\)-structure.
\begin{prop}
	The Euler--Lagrange equation of~\eqref{eq:AFSRS} for \(\Phi\) is
	\begin{equation}\label{eq:EL}
		0 = \Delta^\cD \Phi = \varepsilon^{\alpha\beta}\nabla_{F_\alpha} F_\beta\Phi + \varepsilon^{\alpha\beta}\left(\Div F_\alpha\right)F_\beta\Phi
	\end{equation}
	We will call the differential operator \(\Delta^\cD\), defined here, the \(\cD\)-Laplace operator.
\end{prop}
\begin{proof}
	Let \(\Phi_t: M\times\bbR\to N\) be a perturbation of \(\Phi_0=\Phi\).
	One can expand \(\Phi_t\) in \(t\) around \(0\) and obtains
	\begin{equation}
		\Phi_t=\Phi_0+t\partial_\alpha\Phi_t|_{t=0}+O(t^2)
	\end{equation}
	Let us denote \(\partial_t\Phi_t|_{t=0}=\Xi\in\Gamma(\Phi^*TN)\) and expand \(A\) in \(t\) around \(0\):
	\begin{multline}
		\left.\frac{\d}{\d t}\right|_{t=0}A(\Phi_t, F_A) = \frac12\left.\frac{\d}{\d t}\right|_{t=0}\int_M \varepsilon^{\alpha\beta}\langle F_\alpha\Phi_t, F_\beta\Phi_t\rangle [F^1F^2F^3F^4] \displaybreak[0]\\
		= \frac12\left.\int_M\partial_t\varepsilon^{\alpha\beta}\langle F_\alpha\Phi_t, F_\beta\Phi_t\rangle [F^1F^2F^3F^4] \right|_{t=0} \displaybreak[0]\\
		= \left.\int_M\varepsilon^{\alpha\beta}\langle \nabla^{\Phi_t^*TN}_{\partial_t}F_\alpha\Phi_t, F_\beta\Phi_t\rangle [F^1F^2F^3F^4] \right|_{t=0} \displaybreak[0]\\
		= \left.\int_M\varepsilon^{\alpha\beta}\langle \nabla^{\Phi_t^*TN}_{F_\alpha}\partial_t\Phi_t, F_\beta\Phi_t\rangle [F^1F^2F^3F^4] \right|_{t=0} \displaybreak[0]\\
		= \int_M\varepsilon^{\alpha\beta} \langle\nabla^{\Phi^*TN}_{F_\alpha}\Xi, F_\beta\Phi\rangle [F^1F^2F^3F^4] \displaybreak[0]\\
		= -\int_M\varepsilon^{\alpha\beta}\left(\langle\Xi, \nabla^{\Phi^*TN}_{F_\alpha}F_\beta\Phi\rangle [F^1F^2F^3F^4] - \langle\Xi, F_\beta\Phi\rangle \Lie_{F_\alpha}[F^1F^2F^3F^4] \right)
	\end{multline}
	With the definition of divergence
	\begin{equation}
		\Lie_{F_\alpha} [F^1 F^2 F^3 F^4] = \left(\Div F_\alpha\right) [F^1F^2F^3F^4]
	\end{equation}
	the result follows.
	Of course the Euler--Lagrange equation~\eqref{eq:EL} is \(G\)-invariant like the action~\eqref{eq:AFSRS}.
	The \(\cD\)-Laplace, however, is only \(\UGL(1)\)-invariant.
\end{proof}

We now turn to the question how the action~\eqref{eq:AFSRS} can be represented on an underlying even manifold \(i\colon \Smooth{M}\to M\).
\begin{defn}
\label{defn:CompFields}
	Let \(\Phi\colon M\to N\) be a morphism and \(i\colon \Smooth{M}\to M\) be an underlying even manifold.
	We call the fields
	\begin{align}
		\varphi\colon \Smooth{M}&\to N & \psi\colon \Smooth{M}&\to S^\vee\otimes \varphi^*TN & F\colon \Smooth{M}&\to \varphi^*TN \\
		\varphi &= \Phi\circ i & \psi &= s^\alpha \otimes i^* F_\alpha\Phi & F &= \frac12i^*\Delta^\cD\Phi
	\end{align}
	component fields of \(\Phi\).
	Recall that \(s^\alpha\) is the dual basis to the basis \(s_\alpha = i^*F_\alpha\) of the spinor bundle \(S = i^*\cD\) on \(\Smooth{M}\).
\end{defn}
\begin{rem}
	Suppose that \(X^A = (x^a, \eta^\alpha)\) are Wess--Zumino coordinates for the Wess--Zumino frame \(F_A\).
	Let furthermore \(Y^B\) be local coordinates on \(N\).
	The map \(\Phi\colon M\to N\) is then given by the functions
	\begin{equation}
		\Phi^\#Y^B = f_0^B + \eta^\mu f_\mu^B + \eta^2\eta^1 f_{21}^B
	\end{equation}
	It holds that \(f_0^B = \varphi^\#Y^B\) because \(i^\#\eta^\mu = 0\).
	By the properties of Wess--Zumino coordinates we have that \(i^*F_\alpha = i^*\partial_{\eta^\alpha}\) and thus \(f_\mu^B = \psi_\mu Y^B\).
	Here \(\psi_\mu\) is the coefficient of \(\psi\) in the basis \(s^\mu\) and consequently a derivation on \(\cO_N\) with values in \(\cO_M\).
	If the target manifold \(N=\bbR^p\) is Euclidean space one can show that \(i^*\Delta^\cD = 2i^*\partial_{\eta^1}\partial_{\eta^2}\).
	Consequently the map \(\Phi\) can be written schematically as
	\begin{equation}
		\Phi = \varphi + \eta^\mu \psi_\mu + \eta^2\eta^1 F.
	\end{equation}
\end{rem}

\begin{thm}
\label{thm:AF}
	Let \(M\) be a fiberwise compact family of super Riemann surfaces and \(i\colon \Smooth{M}\to M\) an underlying even manifold.
	We denote by \(g\), \(\chi\), and \(g_S\) respectively the metric, gravitino, and spinor metric on \(\Smooth{M}\) constructed in Section~\ref{sec:SRS} for a given \(\UGL(1)\)-structure on \(M\).
	Let \(\Phi\colon M\to N\) be a morphism to a Riemannian super manifold \((N,n)\) and \(\varphi\), \(\psi\), and \(F\) its component fields, as introduced in Definition~\ref{defn:CompFields}.
	It holds
	\begin{multline}\label{eq:AFRed}
		A(M, \Phi) = A(\varphi, g, \psi, \chi, F) = \int_{\Smooth{M}} \left(\vphantom{\frac12}\| \d\varphi\|^2_{g^\vee\otimes \varphi^*n} + \langle \psi, \Dirac\psi\rangle_{g_S^\vee\otimes \varphi^*n} - \langle F, F\rangle_{\varphi^*n} \right.\\
		+ 2\langle \chi_a\gamma^b\gamma^a\partial_{x^b}\varphi,\psi\rangle_{g_S^\vee\otimes\varphi^*n} + \frac{1}{2}\langle\chi_a, \gamma^b\gamma^a\chi_b\rangle_{g_S}\langle\psi,\psi\rangle_{g_S^\vee\otimes\varphi^*n} \\
		\left.+ \frac16\varepsilon^{\alpha\beta}\varepsilon^{\gamma\delta}\langle R^{\varphi^*TN}(\psi_\alpha, \psi_\gamma)\psi_\delta, \psi_\beta\rangle_{\varphi^*n}\right) \d{vol_g}
	\end{multline}
\end{thm}
The idea for the proof of Theorem~\ref{thm:AF} is Lemma~\ref{prop:Integration}.
One uses crucially that integration in the odd directions is locally a derivation.
In Wess--Zumino coordinates \((x^a, \eta^\alpha)\) for \(F_A\) a local expression for the action is given by
\begin{multline}
	A(M, \Phi) = \frac12\int_M \varepsilon^{\alpha\beta} \langle F_\alpha\Phi, F_\beta\Phi\rangle_{\Phi^*n} {\left(\Ber F\right)}^{-1} [\d{x^1}\d{x^2}\d\eta^1\d\eta^2]\\
	= \frac12\int_{\Smooth{M}} i^*\partial_{\eta^1}\partial_{\eta^2}\left(\varepsilon^{\alpha\beta} \langle F_\alpha\Phi, F_\beta\Phi\rangle_{\Phi^*n} {\left(\Ber F\right)}^{-1}\right) \d{x^1}\d{x^2} \\
	= \frac14\int_{\Smooth{M}} i^*\varepsilon^{\mu\nu}F_\mu F_\nu\left(\varepsilon^{\alpha\beta} \langle F_\alpha\Phi, F_\beta\Phi\rangle_{\Phi^*n} {\left(\Ber F\right)}^{-1}\right) \d{x^1}\d{x^2}
\end{multline}
The expansion of the last expression is given in terms of component fields of \(\Phi\) (compare Definition~\ref{defn:CompFields}), and commutators of \(F_\alpha\) and derivatives of \(\Ber F\).
By Lemma~\ref{lemma:WZPair} the coordinate expansion of \(F_\alpha\), its commutators and the Berezinian are determined by \(g\) and \(\chi\).
The full calculation can be found in~\cite{EK-DR}.

It is now clear how the different symmetries of the action functional~\eqref{AF} arise.
Different \(\UGL(1)\)-reductions of the given \(G\)-structure on \(M\) induce metrics and gravitinos on \(\Smooth{M}\) that differ only by Weyl and super Weyl transformations.
The action functional~\eqref{eq:AFSRS} is \(G\)-invariant, and thus in turn the action functional~\eqref{eq:AFRed} is conformally and super Weyl invariant.
The action functional~\eqref{eq:AFSRS} is formulated without any reference to an embedding of an underlying even manifold, but Theorem~\ref{thm:AF} is.
The independence of~\eqref{eq:AFSRS} of the embedding \(i\) translates into super symmetry of~\eqref{eq:AFRed}.
\begin{prop}
	The Euler--Lagrange equations of the action functional~\eqref{eq:AFRed} are given by the components of the Euler--Lagrange equation of~\eqref{eq:AFSRS}:
	\begin{align}\label{eq:ELComp}
		0 &= i^*\Delta^{\cD} \Phi & 0 &= s^\alpha\otimes i^*\nabla_{F_\alpha} \Delta^\cD \Phi & 0 &= i^*\Delta^\cD \Delta^\cD\Phi
	\end{align}
\end{prop}
\begin{proof}[Sketch of proof]
	Schematically the infinitesimal variation \(\Xi\) of \(\Phi\) can be decomposed
	\begin{equation}
		\Xi = \delta\varphi + \eta^\mu \delta\psi_\mu + \eta^2\eta^1 \delta F.
	\end{equation}
	The infinitesimal variation of the action is then given by
	\begin{equation}
		\delta A = -\int_M \langle \Xi, \Delta^\cD \Phi\rangle [\d{vol_m}].
	\end{equation}
	Integration over the odd variables selects the coefficients of highest degree in \(\eta\), so that
	\begin{equation}
		\begin{split}
			\delta A = -\int_{\Smooth{M}} \frac12\langle \delta\varphi, i^*\Delta^\cD\Delta^\cD \Phi\rangle + \langle \delta \psi, s^\alpha \otimes i^*\nabla_{F_\alpha}\Phi\rangle + \langle \delta F, i^*\Delta^\cD \Phi\rangle \d{vol_g} \qedhere
		\end{split}
	\end{equation}
\end{proof}

By Theorem~\ref{thm:SRSReductionOfStructureGroup} different super Riemann surfaces are given by different \(G\)-structures.
The functional~\eqref{eq:AFSRS} is \(G\)-invariant and different \(G\)-structures lead to different values of the functional.
Consequently the action functional~\eqref{eq:AFSRS} is a functional on the moduli space of super Riemann surfaces for fixed \(\Phi\colon M\to N\).
Unfortunately, the component action functional~\eqref{eq:AFRed} can not be interpreted as a functional on the moduli space of super Riemann surfaces directly.
This is due to the fact that super symmetry is studied up to now only infinitesimally.
We conjectured that super symmetry corresponds to a family of embeddings \(i_t\colon \Smooth{M}_t\to M\).
When expressing~\eqref{eq:AFSRS} as an integral over \(\Smooth{M}_t\) also the domain of integration varies.
In Section~\ref{sec:SRS} we have explained that in order to study the moduli space of super Riemann surfaces, it is necessary to quotient the space of all metrics and gravitinos also by super symmetry.
It is not clear how to take the full quotient by super symmetry, nor how to relate this to the integrals over the different domains of integration \(\Smooth{M}_t\).
However infinitesimal properties of the moduli space of super Riemann surfaces can be studied from~\eqref{eq:AFRed}.
\begin{prop}
	Let \(M\) be a super Riemann surface and \(i\colon \Smooth{M}\to M\) an underlying even manifold.
	By the construction in Section~\ref{sec:SRS}, the geometry of \(M\) is determined by a metric \(g\) and a gravitino \(\chi\) on \(\Smooth{M}\).
	Define the energy-momentum tensor \(T\) of \(A(\varphi, g, \psi, \chi, F)\) by
	\begin{equation}\label{eq:AFT}
		\delta_g A(\varphi, g, \psi, \chi, F) = \int_{\Smooth{M}} \delta g \cdot T \d{vol_g}
	\end{equation}
	and the super current \(J\) by
	\begin{equation}\label{eq:AFJ}
		\delta_\chi A(\varphi, g, \psi, \chi, F) = \int_{\Smooth{M}} \delta \chi \cdot J \d{vol_g}
	\end{equation}
	If the fields \(\varphi\), \(\psi\), and \(F\) fulfil the Euler--Lagrange equations~\eqref{eq:ELComp} the energy-momentum tensor \(T\) is the Noether current  associated to the diffeomorphism invariance, whereas the super current \(J\) is the Noether current to super symmetry.
	The tensors \(T\) and \(J\) are related from the viewpoint of super geometry because super symmetry is induced by a particular super diffeomorphism.

	Furthermore, as the Noether currents are conserved quantities, they are divergence free.
	Consequently, the energy-momentum tensor \(T\) is a holomorphic quadratic differential and the super current \(J\) is a holomorphic section of \(S^\vee\otimes S^\vee\otimes S^\vee\).

	Geometrically, the integrals~\eqref{eq:AFT} and~\eqref{eq:AFJ} can be viewed as cotangent vectors of the moduli space of super Riemann surfaces at \(M\).
\end{prop}

Similar to the case of Riemann surfaces and the harmonic action functional we hope that the action functional~\eqref{eq:AFSRS} may be helpful to derive further results about the moduli space of super Riemann surfaces.
 \addsec{Summary}
We have established the relation between the super symmetric action functional~\eqref{AF} and super Riemann surfaces.
That is, we have shown that for a particular underlying even manifold \(\Smooth{M}\) of the super Riemann surface \(M\) the integral \(A(M,\Phi)\) reduces to the action functional \(A(\varphi, g, \psi, \chi, F)\) on \(\Smooth{M}\).

The first step was to define the underlying family of even manifolds \(\Smooth{M}\to B\) of a family of super manifolds \(M\to B\).
The underlying even manifolds is in between the super manifold \(M\) and the completely reduced space of \(M\), as it still involves odd functions from the base \(B\).

With the help of the underlying even manifold \(\Smooth{M}\) we were able to show that the structure of a super Riemann surface \(M\) is completely determined by an underlying even manifold \(\Smooth{M}\) together with a metric \(g\), a spinor bundle \(S\) and a spinor valued differential form \(\chi\), called gravitino.
The redundancy in the choice of \(g\), \(S\), and \(\chi\) could be shown to coincide with the conformal, super Weyl and super symmetry invariance of the action \(A(\varphi, g, \psi, \chi, F)\).
Infinitesimal deformations of the super Riemann surface can be expressed via infinitesimal deformations of \(g\) and \(\chi\), reproducing the classical result that even infinitesimal deformations of \(M\) are given by holomorphic sections of \(T^\vee M\otimes T^\vee M\), whereas odd infinitesimal deformations are given by holomorphic sections of \({\left(S^\vee\right)}^{\otimes 3}\).

As an outlook, the striking similarities of \(A(M,\Phi)\) with the functional of harmonic maps on Riemann surfaces, together with the results presented in this paper, give rise to the hope that the action functional \(A(M, \Phi)\) and its critical points may be useful to study the moduli space of super Riemann surfaces.
On one hand, the geometry of super Riemann surfaces and their moduli involve the integrability conditions~\eqref{eq:FrameIntegrabilityConditions}.
On the other hand, however, the characterization of super Riemann surfaces in terms of metrics and gravitinos is not obstructed.
Due to Theorem~\ref{thm:AF}, the action functional \(A(\varphi, g, \psi, \chi, F)\) in terms of metric and gravitino is well defined on the moduli space of super Riemann surfaces.
 
\addsec{Acknowledgement}
We wish to thank Ron Donagi for helpful comments on earlier versions of this paper.
The second author wants to thank the International Max Planck Research School Mathematics in the Sciences for financial support.
The research leading to these results has received funding from the European Research Council under the European Union's Seventh Framework Programme (FP7/2007--2013) / ERC grant agreement nº 267087.

\printbibliography

\end{document}